 \documentclass[final,3p,times]{elsarticle}

 \usepackage{graphicx}
\usepackage{color}
\usepackage{caption}
\usepackage{subcaption}
\usepackage{amssymb}
 \usepackage{amsthm}

\newtheorem{thm}{Theorem}[section]

\newtheorem{prop}[thm]{Proposition}
\theoremstyle{definition}
\newtheorem{defn}[thm]{Definition}
\theoremstyle{remark}

\newtheorem*{ex}{Example}

\begin{document}

\begin{frontmatter}

\title{Combinatorics of Linked Systems of Quartet Trees}

\author[label1]{Emili Moan}
\address[label1]{Dept.\ of Mathematics, Winthrop University, Rock Hill, SC 29733 USA}

\cortext[cor1]{corresponding author}

\ead{pricee4@winthrop.edu}

\author[label1]{Joseph Rusinko\corref{cor1}}
\ead{rusinkoj@winthrop.edu}


\begin{abstract}
We apply classical quartet techniques to the problem of phylogenetic decisiveness and find a value $k$ such that all collections of at least $k$ quartets are decisive. Moreover, we prove that this bound is optimal and give a lower-bound on the probability that a collection of quartets is decisive.
\end{abstract}


\end{frontmatter}


\section{Overview}
Evolutionary biologists represent relationships between groups of organisms with phylogenetic trees. Supertree methods were designed to handle the computationally difficult problem of reconstructing such trees for large data sets. Those methods generate a group of accurate, smaller input trees and combine them into a single supertree. Four-taxa trees, known as quartet trees, are commonly used as inputs in supertree methods. 

Most quartet amalgamation algorithms use all quartet trees generated from sequencing data or only remove quartet trees that appear to be incorrect. As quartet trees may contain overlapping information, it is possible that a smaller number of trees may provide sufficient information for accurate reconstruction. 

In \cite{bocker1998}, B{\"o}cker, et al. developed a sufficient condition for a set of quartet trees to be definitive.  For any tree on a taxon set $[n]=\{1,2,\cdots, n\}$, we develop a system of quartet trees that meets B{\"o}cker’s criteria, known as a linked system. Additionally, we develop collections of linked systems known as meshed systems.

Recently, Steel and Sanderson asked for which collections of sets of taxa do the corresponding induced subtrees determine a unique supertree. They called such collections decisive. The notion of decisiveness can be viewed as a generalization of definitiveness where no information is required about the particular subtrees that the subsets of taxa induce.  This notion plays an important role in supertree reconstruction since it a priori addresses the question about which subsets of taxa must be analyzed to ensure that a unique supertree can be reconstructed.

We use the term quartet to refer to any four element taxon subset, and the term quartet tree when referencing a resolved four taxa tree. Using meshed systems, we find a minimal number $k(n)$ such that every collection of at least $k$ quartets is decisive. We use this number to find a lower bound on the probability that an arbitrary collection of quartets is decisive. 

Finally, we find that meshed systems may be useful in amalgamation algorithms, such as Maxcut \cite{snir}, that do not always find the correct supertree when given a definitive system of quartet trees.

\section{Linked Systems}

We adopt the terminology in \cite{huber}, except in noted instances when we follow \cite{semple} or \cite{steel}.  Phylogenetic trees display relationships among a finite set of taxonomic units. 

\begin{defn}
A \emph{binary phylogenetic tree}, $T=(V,E,\varphi)$ on a finite set of taxa $X$, is a triple consisting of a finite set of vertices, V , a set $E$ of edges between vertices, and a "labeling" map $\varphi : X \rightarrow L$, where $L \subset V$ contains all vertices of degree one or \emph{leaves}, such that the graph $(V,E)$ is an unrooted binary tree and
the map $\varphi$ induces a bijection between $X$ and the set $L$ of leaves of $T$.
\end{defn}

An edge that contains a leaf is an \emph{exterior edge}. The non-leaf vertex of an exterior edge is  the \emph{internal vertex of $e$}, denoted $v_{int}(e)$. Two exterior edges sharing an internal vertex form a \emph{cherry}. Any edge that is not an exterior edge is an \emph{interior edge}.

While edge length plays an important role in phylogenetics, we do not take it into account, and adopt instead a topological definition of tree isomorphism.

\begin{defn}
Phylogenetic trees, $T_1=(V_1, E_1, \varphi_1)$ and  $T_2=(V_2, E_2, \varphi_2)$ on a taxon set $X$, are \emph{isomorphic} if there exists a bijective map $f : V_1 \rightarrow V_2$, called an \emph{isomorphism}, such that if $\{u,v\} \in E_1$ then $\{f (u), f (v)\} \in E_2$ and  for every  $x \in X$ we have $\varphi_2 (x)=f (\varphi_1 (x))$ . 
\end{defn}

It is impossible to distinguish phylogenetic relationships from unrooted trees with fewer than four taxa; thus, supertree reconstruction algorithms frequently use four taxa trees or \emph{quartets trees} as inputs \cite{snir, warnow, strimmer}. \emph{Quartet trees} are binary phylogenetic trees on four leaves. Such trees are in one-to-one correspondence with two-element subsets of $X$ such as $\{\{a, b\},\{c,d\}\}$ according to the separation of the four leaves by the interior edge. The union of all four taxa is the \emph{support of $q$}, denoted \emph{$supp(q)$}. 

Quartet trees contain an interior edge which separates the taxa into two pairs. Similarly, removing an interior edge of a tree separates the graph into two connected components. An edge $e$ \emph{separates} taxa $a$ and $b$ from $c$ and $d$ if $\{a,b\}$ and $\{c,d\}$ are subsets of the vertex sets of different connected components of $T-\{e\}$. This separation points to a relationship between edges of a tree and quartet trees. A quartet tree $ab|cd$ is \emph{displayed by a binary phylogenetic tree $T$} if there exists an edge $e \in E$ that separates $a$ and $b$ from $c$ and $d$.  

Denote the set of all quartet trees on a taxon set $X$ by $Q(X)$. Any subset $Q$ of $Q(X)$ is called a \emph{system of quartet trees} on $X$ with the support defined by $supp(Q)=\displaystyle\bigcup_{q\in Q} supp(q)$. Additionally, we denote the set of all quartet trees displayed by a tree $T$ by $Q_T$. A system of quartet trees $Q$ is \emph{compatible} if there exists a tree $T$ such that $Q \subseteq Q_T$. 

\begin{defn}
(\cite{semple}) Let $T = (V,E,\varphi)$ be a binary phylogenetic tree and let $ab|cd \in Q_T$. An interior edge $e$ of $T$ is \emph{distinguished} by $ab|cd$ if $e$ is the only edge that separates $a$ and $b$ from $c$ and $d$.  
\end{defn}

Quartet trees which distinguish edges are a powerful input to quartet amalgamation algorithms. These algorithms must handle non-compatible systems of quartet trees.  However, even compatible systems may be difficult to resolve as multiple trees may display a particular collection of quartet trees.

\begin{defn} (\protect\cite{steel})
A system of quartet trees, $Q$, is \emph{definitive}, if up to isomorphism, there exists a unique binary phylogenetic tree $T$ for which $Q \subseteq Q_T$. 
\end{defn}

In \cite{bocker1998}, B{\"o}cker described various criteria for a system of quartet trees of the size $n-3$ to be definitive. We construct systems of quartet trees that meet this criteria and make note of some useful applications of these systems. 

\begin{prop} 
(Example 3.7 of \cite{bocker1998})
If $T$ is a binary tree such that the interior edges of $T$ are labeled $E=\{e_1,...,e_{n-3}\}$, and $Q$ is a system of quartet trees such that each $q_i \in Q$ distinguishes a unique edge $e_i$ in $T$ with \[|supp(q_i)\backslash \bigcup_{j<i} supp(q_j)|= 1\] for $i=2,...,n-3$, then $Q$ is definitive. 
\label{bocker2}
\end{prop}

We create a system of quartet trees that satisfies the hypotheses of Proposition \ref{bocker2}, known as a linked system, by imposing an ordering on the interior edges of a tree and the quartet trees which distinguish those edges. We define linked systems in terms of the associated graph. 

\begin{defn}
For a compatible system of $n-3$ quartet trees $Q$ on a taxon set $X$, define the associated graph $G_T(Q)$ with vertex set $V$ and edge set $E$ as follows:
\begin{itemize}
\item The vertex set $V$ is the set of all quartet trees $q \in Q$ which distinguish a unique edge in $T$. 
\item Vertex pairs $\{q_i,q_j\}$ are connected by an edge $e \in E$ if the edge $e_i$ that $q_i$ distinguishes is adjacent to the edge $e_j$ that $q_j$ distinguishes and $|supp(\{q_i,q_j\})|=5$.
\end{itemize}
\label{order}
\end{defn}

\begin{defn}
Two quartet trees are \emph{linked} if their vertices are connected in $G_T(Q)$. The system of quartet trees $Q$, is \emph{a linked system} if $G_T(Q)$ is connected. See Figure \protect\ref{fig:graph} for an example.
\end{defn}

\begin{figure}[ht]
 \centering
 \includegraphics[width=.4\linewidth]{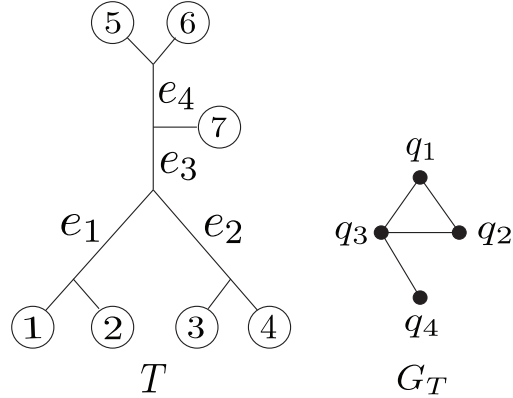}
 \caption{A binary phylogenetic tree $T$ and the associated graph $G_T(Q)$ for the quartet trees $q_1=12|35$, $q_2=34|15$, $q_3=57|13$ and $q_4=56|71$.}
\label{fig:graph}
\end{figure}

In Section \protect\ref{sec:linked} we prove that linked systems are definitive. Linked systems also help illuminate the broader concept of phylogenetic decisiveness which we review here.

For a binary phylogenetic tree $T$ and a subset $Y$ of $X$, let $T|Y$ denote the induced binary phylogenetic tree on leaf set $Y$ (the tree obtained from the minimal subtree connecting $Y$ by suppressing any vertices of degree $2$). Let $\mathcal{S}$ be the collection of subsets of a set $X$ of size four, we refer to all such subsets as \emph{quartets}. 
\cite{steel}

\begin{defn} \cite{steel}
We say that $\mathcal{S}$ is \emph{phylogenetically decisive} if it satisfies the following property: If $T$ and $T$' are binary phylogenetic trees, with $T|Y=T$'$|Y$ for all $Y \in \mathcal{S}$, then $T=T$'. 
\end{defn}

We will use collections of linked systems to find the minimal number $k(n)$ such that a collection of at least $k$ quartets is phylogenetically decisive.

\section{Applications of Linked Systems}
\label{sec:linked}
We first show that linked systems meet B{\"o}cker's criteria for defining a unique tree. 

\begin{thm}
 Every linked system of quartet trees is definitive.
\label{hierarchy} 
\end{thm}

\begin{proof}
Let $Q$ be a  linked system of quartet trees on a tree $T$ on a taxon set $X$. Linked systems are of size $n-3$ and each quartet tree distinguishes a unique edge in $T$. Let $T$ be a binary phylogenetic tree on a taxon set $X$ and let $e_1$ be an interior edge adjacent to a cherry. The tree is connected, which implies we can label the remaining interior edges $\{e_2, \cdots, e_{n-3} \}$ such that $e_j$ is adjacent to some $e_i$ with $i <j$. Moreover, because the support of each pair of quartet trees $\{q_i,q_j\}$ that distinguishes adjacent edges $\{e_i,e_j\}$ is five, each pair of quartet trees shares three taxa and each additional quartet tree in $Q$ adds one new taxon to the support of $Q$. Thus, linked systems meet the criteria in Proposition \ref{bocker2} and are definitive. 
\end{proof}

Though all linked systems are definitive, we find that not all definitive systems of size $n-3$ are linked.

\begin{ex}
The system of quartet trees $Q=\{12|36,23|45,24|56\}$ meets the criteria established in Proposition \ref{bocker2}, and is thus definitive. The graph $G_T(Q)$ contains the three vertices $q_1$, $q_2$, and $q_3$, where $q_2$ and $q_3$ are connected by an edge and $q_1$ is an isolated vertex. Thus, $Q$ is not a linked system. 
\end{ex}

Since B{\"o}cker's system and linked systems contain $n-3$ quartet trees, one might surmise that all compatible systems of quartet trees of a modest size would be definitive. However, there are  large systems of compatible quartet trees which are not definitive and large collections of quartets which are not decisive. A collection of quartets and the induced quartet trees on a caterpillar tree provides one such example.

\begin{defn} (\cite{semple}) A \emph{caterpillar} on $n$ leaves is a binary phylogenetic tree for which there exists an induced subtree on a sequence of distinct interior vertices $v_1,v_2,...,v_k$ such that, for all $i \in \{1,2,...,k-1\}$, $v_i$ and $v_{i+1}$ are adjacent.
\end{defn}

The ordering of vertices in a caterpillar tree induces an ordering of interior edges $e_i$ where $e_i$ connects $v_i$ with $v_{i+1}$.  We use this ordering to construct large families of quartet trees shared by several caterpillar trees.

\begin{thm}
The minimal number $k(n)$ such that every collection of quartets $S$  with $|S| \geq k$ is decisive is greater than ${n \choose 4} - (n-3)$.
\label{ge}
\end{thm}
\begin{proof} 
Let \{a,b,c\} be a subset of $X$. We define $T_1$, $T_2$, and $T_3$ to be three distinct caterpillar trees of size $n \geq 4$ that differ only in the placement of three taxa $a$, $b$, and $c$, such that for each tree $v_{int}(a), v_{int}(b)$ and $v_{int}(c)$ are incident to $e_1$. Denote $\mathcal{S}$ the set of $n-3$ sets $\displaystyle\bigcup_{i=1}^{n-3}Y_i=\{a,b,c,y|y \in X-\{a,b,c\}\}$ and let $\mathcal{S'}$ be the complement of $\mathcal{S}$. We observe that for all $Y\in \mathcal{S'}$ we have $T_1|Y=T_2|Y=T_3|Y$, but $T_1 \neq T_2 \neq T_3$.  Therefore, $S'$ is not decisive.  Since $|\mathcal{S'}|={n \choose 4} - (n-3)$ the minimal number $k(n)$ such that every collection $\mathcal{S}$ of quartets with $|\mathcal{S}| \geq k$ is decisive is greater than ${n \choose 4} - (n-3)$.
\end{proof}

To show that that sets of quartets of size ${n \choose 4} - (n-3)$ are decisive, we prove that $Q_T$  contains at least $n-3$ disjoint linked systems, ensuring the removal of any $n-4$ quartet trees from a compatible system would leave at least one linked system. We introduce a process for building such systems by using a seed quartet tree which distinguishes an edge of a tree, and systematically constructing additional quartet trees which distinguish the same edge.

In a phylogenetic tree, each interior edge $e=(v_l,v_r)$ is adjacent to four edges $e_i$, $e_j$, $e_h$, and $e_k$, which divide the tree into four components and partition the set of taxa $X$ into four distinct sets $A_i$, $A_j$, $A_k$, and $A_h$ with $x \in A_n$ if the unique path from $x$ to $v_l$ contains the edge $e_n$. 

\begin{defn} Let $q=ij|kh$ be a quartet tree that distinguishes an edge $e$ and let $i \in A_i$, $j \in A_j$, $k \in A_k$, and $h \in A_h$, where $A_i$, $A_j$, $A_k$, and $A_h$ are partitions of $X$ induced by $e$. For $x\in X - supp(q)$  define the \emph{quartet tree substitution} $q(x)$ to be the unique quartet tree in which the taxon $x \in A_n$ replaces the taxon in $q$ that is in $supp(q) \cap A_n$.
\end{defn}

Notice $q(x)$ and $q$ must distinguish the same edge of the tree.

\begin{defn}
Let the quartet tree $q$  distinguish an edge $e$ of a tree. Define the \emph{vine of q}  by $v(q)=\{q\} \cup \displaystyle\bigcup_{x\in X-supp(q)}q(x)$.  We refer to $q$ as the \emph{seed} of the vine.
\end{defn}

The following shows that if two quartet trees are linked, then so are their vines.

\begin{defn}
Two vines $v(q_i)$ and $v(q_j)$ are \emph{linked} if for each $q_i \in v(q_i)$ there exists a unique $q_j \in v(q_j)$ such that $q_i$ and $q_j$ are linked. 
\end{defn}

\begin{thm}
If $q_i$ and $q_j$ are linked quartet trees, then  the associated  vines $v(q_i)$ and $v(q_j)$ are linked.
\label{linkedquartets}
\end{thm}

\begin{proof}
Assume that $\{q_i,q_j\}$ are the seeds of the adjacent edges $e_i$ and $e_j$ and are linked in $T$. Let $v(q_i)$ and $v(q_j)$ be the associated vines. 

Because the $supp(q_i,q_j)=5$, each quartet tree contains one taxon that the other does not. Let $z$ be the taxon in $supp(q_j)-supp(q_i)$ and $y$ be the taxon in $supp(q_i)-supp(q_j)$. Use quartet tree substitution to construct the quartet trees $q_i(z)$ and $q_j(y)$. By construction, $supp(q_i,q_j)=supp(q_i(z),q_j(y))$ and $q_i(z)$ and $q_j(y)$ are linked.

Use quartet tree substitution with each remaining taxon $x \in X-supp(q_i,q_j)$ on $q_i$ and $q_j$ to construct the remaining quartet trees in $v(q_i)$ and $v(q_j)$. In the construction of each $\{q_i(x),q_j(x)\}$ one taxon (Case 1) or two taxa (Case 2) are removed and one taxon $x$ is introduced. Thus, for  $x \in X - (supp(q_i,q_j))$ we have $4 \leq |supp(q_i(x),q_j(x))| \leq 6$. 

In Case 1, $x$ replaces the same taxon in $q_i$ and $q_j$ and $supp(q_i,q_j)$ remains the same. 

In Case 2, $x$ replaces one taxon in $q_i$ and a different taxon in $q_j$. 

Assume that $x$ replaces two different taxa in $supp(q_i,q_j)-supp(q_i \cap q_j)$. Then, $|supp(q_i(x),q_j(x))|=4$ and $q_i(x)=q_j(x)$. This is not possible since $q_i$ and $q_j$ distinguish different edges. Thus, $x$ does not replace two different taxa in $supp(q_i,q_j)-supp(q_i \cap q_j)$ and $|supp(q_i(x),q_j(x))| \neq 4$.

Now assume that $x$ replaces two different taxa in $supp(q_i \cap q_j)$. Then, $|supp(q_i(x),q_j(x))|=6$. Recall that $q_i$ and $q_j$ distinguish the edges $e_i$ and $e_j$. Thus, in order for $x$ to replace two different taxa in $supp(q_i \cap q_j)$, $x$ would have to be in two different sets of the partition that $e_i$ induces on $X$. This is not possible because $x$ cannot be in two different sets of a partition. Thus, $x$ does not replace two different taxa in $supp(q_i \cap q_j)$ and $|supp(q_i(x),q_j(x))| \neq 6$. 

Thus, in this case, $x$ replaces one taxon in $supp(q_i \cap q_j)$ and one taxon in $supp(q_i,q_j)-supp(q_i \cap q_j)$ and $|supp(\{q_i(x),q_j(x)\})|=5$. Therefore the vines $v(q_i)$ and  $v(q_j)$ are linked.
\end{proof}

A linking between vines allows us to construct multiple disjoint linked systems of quartet trees. We refer to these systems as \emph{meshed systems} and use them to show that any set of quartets of sufficient size is decisive. 

\begin{defn}
A \emph{meshed system} on a tree $T$ with taxon set $X$ is an $(n-3)$ by $(n-3)$ array of quartet trees, where each row is a linked system and each column is a vine.
\end{defn}

Note that the existence of a meshed system  ensures that the removal of up to $n-4$ quartet trees from $Q_T$ must leave at least one definitive set.

\begin{thm}
For any binary phylogenetic tree $T$ on a taxon set $X$, the system $Q_T$ of all quartet trees displayed by $T$ contains a meshed system. 
\label{number}
\end{thm}

\begin{proof}
Let $T$ be a binary phylogenetic tree on a taxon set $X$ and let $e_1$ be an interior edge adjacent to a cherry. The tree is connected, which implies we can label the remaining interior edges $\{e_2, \cdots, e_{n-3} \}$ such that $e_j$ is adjacent to some $e_i$ with $i <j$.

Let $e_j$ be adjacent to $e_i$ with $i<j$. We know that $e_i$ separates $T$ into two connected components $T_i^a$ and $T_i^b$. Moreover, $e_j$ separates $T$ into two connected components $T_j^a$ and $T_j^b$. Because $e_i$ and $e_j$ are adjacent, $supp(T_i^b) \cap supp(T_j^a) \neq \emptyset$. Let $q_i=ab|cd$ and $q_j=ac|de$ such that $a \in supp(T_i^a)$, $c \in supp(T_j^b)$, and $d \in supp(T_i^b) \cap supp(T_j^a)$. Thus, $q_i$ and $q_j$ are linked for all $e_i$ and $e_j$ in $T$. Thus, we have a linked system that makes up the first row of our matrix. 

Using quartet tree substitution, construct vines $v(q_i)$ and $v(q_j)$. By Theorem \ref{linkedquartets}  the vines  $v(q_j)$ and $v(q_i)$ are linked. Thus, we have $n-3$ disjoint columns of quartet trees in our matrix. Additionally, for each pair of linked quartet trees $\{q_i,q_j\}$ in row one, there exists a pair $\{q_i(x),q_j(x)\}$ in the remaining rows of the matrix that are linked. Thus, we have $n-3$ rows of linked systems.  

Therefore, $Q_T$ contains a meshed system. 
\end{proof}

The existence of a meshed system allows us to find the minimal number, $k(n)$, such that every collection $\mathcal{S}$, of quartets with $|\mathcal{S}| \geq k$ is decisive.

\begin{thm}
\label{thm:big}
The  number $k(n)={n \choose 4} - (n-4)$ is the smallest number such that every collection of quartets $\mathcal{S}$ on a taxon set $X=[n]$ such that  $\mathcal{S} \geq k$ is decisive.
\label{minimal}
\end{thm}

\begin{proof}
Let $\mathcal{S}$ be a collection of quartets on a taxon set $X=[n]$, with $|\mathcal{S}| \ge {n \choose 4} - (n-4)$. Let $T$ and $T'$ be two phylogenetic trees such that $T|Y=T$'$|Y$ for all $Y \in \mathcal{S}$. We define $Q \subset Q_T$ to be the collection of $T|Y$ for all $Y \in \mathcal{S}$. By Theorem \ref{number}, $Q_T$ contains a meshed system $M$. By the pigeon hole principle, if  $|Q| \geq {n \choose 4} - (n-4)$ then $Q$ must contain one of the linked systems in $M$, and by Theorem \ref{hierarchy}, $Q$ is definitive.  Thus, $T$ is the unique tree which displays $Q$.  However since $Q$ is also $T'|Y$ for all $Y \in \mathcal{S}$, we must have $T=T'$.  Therefore, $\mathcal{S}$ is decisive. Moreover, Theorem \ref{ge} shows that $k \geq {n \choose 4} - (n-4)$.  Therefore  $k(n)={n \choose 4} - (n-4)$ is the minimal number such that every collection of quartets with $|\mathcal{S}| \geq k$ is decisive.
\end{proof}

In addition to establishing requirements for collections of subsets of $[n]$ to be decisive, \cite{steel} provides a formula for the probability that a particular collections of subsets of $[n]$ will be decisive for an arbitrarily sampled phylogenetic tree. In this section, we prove a similar result by finding a lower-bound for the probability that a collection of subsets of $[n]$ of a particular size will be phylogenetically decisive. This bound is independent of the underlying tree topology.

\begin{thm}
The probability $p(X,k)$ that an arbitrary collection of $k$ quartets is decisive has the property \[p(X,k) \geq \frac{\sum\limits_{i=1}^{n-3} (-1)^{i+1}{{n-3} \choose i}{{|Q_T|-i(n-3)} \choose {|Q_T|-k}}}{{|Q_T| \choose {k}}}.\]
\end{thm}

\begin{proof}
 Let $\mathcal{S}$ be a collection of $k$ quartets.  Let $T$ and $T'$ be two phylogenetic trees such that $T|Y=T'|Y$ for all $Y \in \mathcal{S}$. We define $Q \subset Q_T$ to be the collection of $T|Y$ for all $Y \in \mathcal{S}$.  Following Theorem \protect\ref{thm:big}, if $Q$ contains a definitive set of quartet trees, then $S$ is decisive.  By Theorem \protect\ref{hierarchy}, if a collection of compatible quartet trees contains a linked system of quartets, then it is definitive. Thus, the probability that a collection $\mathcal{S}$ is decisive is at least the probability that $Q$ contains one of the $n-3$ disjoint linked systems constructed in Theorem \protect\ref{number}. The formula follows from applying the inclusion-exclusion principle to count the number of subsets of size $k$, which contain one of the disjoint systems of linked quartets. 
\end{proof}

To illustrate the utility of the formula, we express the lower bound probability versus the number of quartets selected in Figure \ref{fig:probgraph}. In Figure \ref{fig:decisive}, we plot the number quartets required to ensure a fixed accuracy as a power of $n$. Notice the number of quartets needed to ensure the sample is decisive with accuracy of $25 \%$ is on the order of $n^c$ with $c \sim 3.3$ and is almost indistinguishable from the number required to ensure $99\%$ accuracy. 


\begin{figure}[ht]
 \centering
 \includegraphics[width=.65\linewidth]{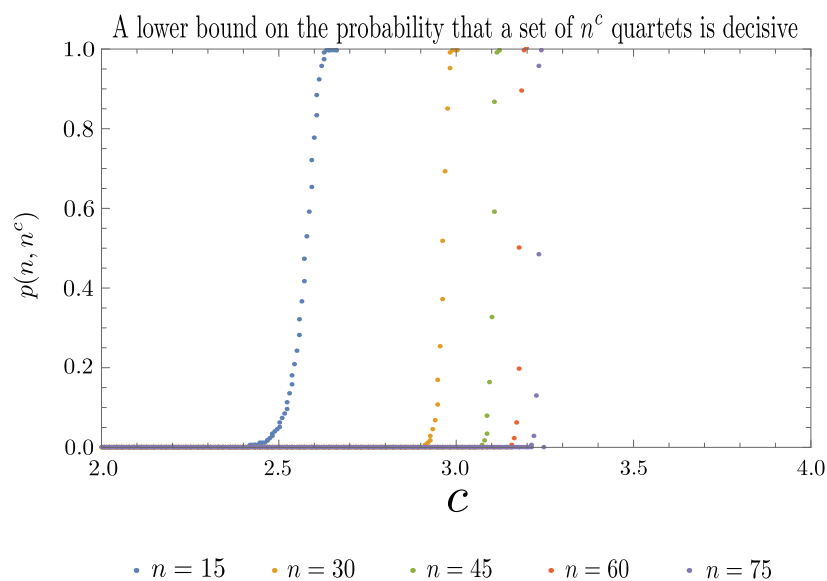}
 \caption{A lower bound on the probability that a set of quartets is decisive.}
\label{fig:probgraph}
\end{figure}

\begin{figure}[ht]
 \centering
 \includegraphics[width=.65\linewidth]{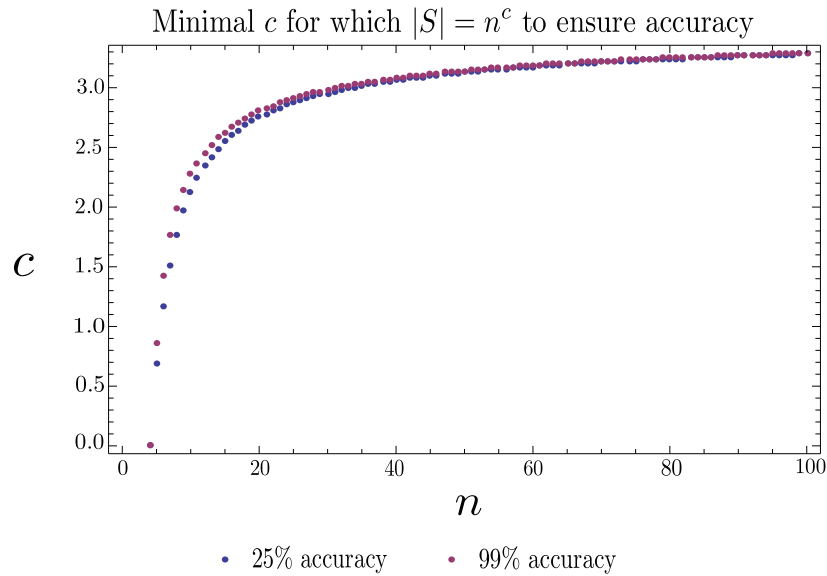}
 \caption{Size of compatible quartets as a power of $n$  required to ensure a decisive subset with fixed probability}
\label{fig:decisive}
\end{figure}

\section{Conclusion}

Using the criteria that B{\"o}cker established in \cite{bocker1998}, we have developed a new type of definitive system of quartet trees, linked systems. We have also developed groups of linked systems, known as meshed systems. We have used meshed systems to show that the number of quartets required to ensure decisiveness is on the order of $O(n^4)$. Moreover, we have used meshed systems to show the probability that an arbitrary collection of quartets contains a decisive system. These results lend credence to sampling  quartets on the order of $n^{3.3}$. 

It has been suggested that smaller sets of representative quartet trees will play a crucial role in developing efficient scalable supertree methods, as the use of all quartet tree samples may be computationally inefficient  \cite{supertree}. Thus, linked systems may be useful inputs in such algorithms. However, some supertree methods such as Quartets MaxCut do not always return a fully resolved tree even when the input sets contain small definitive systems. For example, MaxCut does not return a fully resolved tree for the linked system $Q_1=\{12|35,13|45,14|56\}$, but returns the correct tree for the meshed system $M=\{Q_1,Q_2,Q_3\}$ where $Q_2=\{12|34,23|45,24|56\}$ and $Q_3=\{12|36,23|46,34|56\}$. Therefore, we anticipate that both linked and meshed systems will serve as efficient inputs for future supertree algorithms, as these algorithms could be reformulated to emphasize small definitive units.

\section{Acknowledgements} 
Both authors were supported by grants from the National Center for Research Resources
(5 P20 RR016461) and the National Institute of General Medical Sciences (8 P20 GM103499)
from the National Institutes of Health. We would also like to thank Dr. Mike Steel for introducing us to the concept of decisiveness and giving input throughout the writing process. 
 
\appendix




\bibliographystyle{elsarticle-num}

\bibliography{finalbib}

\end{document}